\documentclass[12pt,a4,english]{article}
\usepackage[latin9]{inputenc}
\usepackage{geometry}
\usepackage{verbatim}
\usepackage{amsmath}
\usepackage{amsfonts}
\usepackage{amsthm}
\usepackage{mathtools}
\usepackage{graphicx}
\usepackage{longtable}
\usepackage{setspace}
\usepackage{upgreek}
\usepackage{comment}
\usepackage[round,sort,authoryear]{natbib}
\usepackage[colorlinks=true,linkcolor=blue, citecolor=blue, urlcolor = blue]{hyperref}

\usepackage{titlesec}
\usepackage{babel}
\usepackage[toc,page]{appendix}
\usepackage{booktabs,caption}

\usepackage[labelfont=bf,textfont=md]{caption}
\usepackage[labelsep=space]{caption}

\usepackage{lmodern}
\usepackage[T1]{fontenc}

\DeclarePairedDelimiter{\floor}{\lfloor}{\rfloor}

\usepackage{bibentry}
\nobibliography*

\usepackage{fancyhdr}
\numberwithin{equation}{section}

\titleformat*{\section}{\large \bfseries}
\titleformat*{\subsection}{\normalsize \bfseries}
\titleformat*{\subsubsection}{\small \bfseries}

\newcommand\norm[1]{\left\lVert#1\right\rVert}

\usepackage{fancyhdr}
\numberwithin{equation}{section}

\theoremstyle{definition}
\newtheorem{theorem}{Theorem}

\newtheorem{assumption}{Assumption}
\newtheorem{lemma}{Lemma}
\newtheorem{example}{Example}

\newtheorem{remark}{Remark}
\newtheorem{corollary}{Corollary}

\usepackage{amsmath}

\begin{document}
\pagenumbering{roman}







\title{ {\Large \textbf{Estimation and Inference in Threshold Predictive Regression Models with Locally Explosive Processes}\thanks{I wish to thank Jean-Yves Pitarakis and Tassos Magdalinos for helpful discussions and constructive feedback. I also wish to thank session participants of the \textit{IX Workshop in Time Series Econometrics} at the University of Zaragoza, as well as the participants of the 2018 and 2019 \textit{Econometric Workshops} at the Department of Economics of the University of Southampton for helpful discussions. } \\
}
}

%

\author{\textbf{Christis Katsouris}\footnote{Lecturer in Economics, Department of Economics, University of Exeter Business School, Exeter  EX4 4PU, United Kingdom. \textit{E-mail Address}: \textcolor{blue}{C.Katsouris@exeter.ac.uk}.} \\ \textit{Department of Economics}, \\  \textit{University of Exeter} \\   \\ \small{ \textcolor{blue}{First Version:}\footnote{The manuscript was previously titled: "\textit{Inference in Threshold Predictive Regression Models with Hybrid Stochastic Local Unit Roots}", when the author was a Ph.D. Candidate in Economics at the Department of Economics of the University of Southampton. Financial support from the VC PhD studentship of the University of Southampton is gratefully acknowledged.} June 2021}  \\ \small{ \textcolor{blue}{This Version:} May 2023}          
      }



\date{}

\maketitle

\begin{abstract}
\vspace*{-0.8 em}
In this paper, we study the estimation of the threshold predictive regression model with hybrid stochastic local unit root predictors. We demonstrate the estimation procedure and derive the asymptotic distribution of the least square estimator and the IV based estimator proposed by \cite{magdalinos2009limit}, under the null hypothesis of a diminishing threshold effect. Simulation experiments focus on the finite sample performance of our proposed estimators and the corresponding predictability tests as in \cite{gonzalo2012regime}, under the presence of threshold effects with stochastic local unit roots. An empirical application to stock return equity indices, illustrate the usefulness of our framework in uncovering regimes of predictability during certain periods. In particular, we focus on an aspect not previously examined in the predictability literature, that is, the effect of economic policy uncertainty. 
\\

\\

\textbf{JEL} classification numbers: C12, C22, C58
\end{abstract}
 

\newpage 

\setcounter{page}{1}
\pagenumbering{arabic}

\section{Introduction}

Consider the linear predictive regression model with the following form 
\begin{align}
y_t = \alpha_0 + \beta_0^{\prime} \boldsymbol{x}_{t-1} + u_{yt} \equiv \boldsymbol{\theta}_0^{\prime} \boldsymbol{X}_{t-1} + u_{yt},  \ \ \text{with} \ t = 1,..., n \ \ \text{and} \ \boldsymbol{x}_0 = 0,
\end{align}
where $y_t \in \mathbb{R}$, $\boldsymbol{X}_{t-1} \in \mathbb{R}^{p+1} = \big( 1, x_{t-1}^{\prime} \big)^{\prime}$ includes both the intercept and the $p-$dimensional regressor vector $x_{t-1}$, with the parameter vector being $\theta_0 = \left( \alpha_0, \beta_0 \right)$ such that $\beta_0$ consists of $p$ coefficients. Moreover, throughout the paper, we assume that regressors are generated via the following locally persistent autoregression process studied by \cite{lieberman2017multivariate, lieberman2018iv, lieberman2020hybrid}
\begin{align}
\label{LSTUR}
\boldsymbol{x}_t = \boldsymbol{R}_{nt} \boldsymbol{x}_{t-1} + \boldsymbol{u}_{xt}, \ \ \text{with} \ \ \boldsymbol{R}_{nt} = 
\left[  \underbrace{\left( \boldsymbol{I}_p +  \displaystyle \frac{ \boldsymbol{C}_p }{n} \right)} + \underbrace{ \left( \displaystyle  \frac{1}{ \sqrt{n} } \langle \varphi , \boldsymbol{u}_{\varphi  t} \rangle \otimes \boldsymbol{I}_p +  \frac{1}{ 2n } \langle \varphi , \boldsymbol{u}_{\varphi  t} \rangle^2 \otimes \boldsymbol{I}_p  \right) } \right] 
\end{align}
where $\varphi$ is an unknown coefficient such that $\boldsymbol{\varphi} := \big( \varphi_1,..., \varphi_d \big)$ with $d \equiv p$, that is, we assume that the parameter $\varphi$ has the same dimension as the number of regressors included in the model. Notice that the inner product in expression \eqref{LSTUR}  implies that $\langle \varphi , u_{\varphi  t} \rangle = \sum_{j=1}^d \varphi_j u_{\varphi, jt} \equiv \widetilde{\boldsymbol{\varphi}}_{( 1 \times n )}$. Although the persistence coefficients, $c_i \ \forall \ i \in \left\{ 1,..., p \right\}$ remain unchanged across $t$, the autocorrelation coefficient is estimable separately for each $t \in \left\{ 1,..., n \right\}$ due to the terms $\langle \varphi , u_{\varphi  t} \rangle \otimes I_p$ and $\langle \varphi , u_{\varphi  t} \rangle^2 \otimes I_p$. Furthermore, the consistent and robust estimation (to the presence of the unknown persistence) of the true parameter vector $\boldsymbol{\theta}_0$ is the main interest of our study. Additionally, the estimation of the parameter vector $\boldsymbol{\varphi}$ is a relevant but challenging  aspect due to its stochastic nature. We focus on estimation and inference aspects in threshold predictive regression models with locally explosive regressors. In particular, the non time-invariant property of the autocorrelation coefficient in the nonstationary autoregressive model implies a time-specific estimation of these matrices
\begin{align*}
\boldsymbol{R}_{1n} 
&= 
\left( \boldsymbol{I}_p +  \displaystyle \frac{\boldsymbol{C}_p}{n} \right) + \left( \displaystyle  \frac{1}{ \sqrt{n} } \langle \boldsymbol{\varphi}, \boldsymbol{u}_{\varphi  1} \rangle \otimes \boldsymbol{I}_p +  \frac{1}{ 2n } \langle \boldsymbol{\varphi}, \boldsymbol{u}_{\varphi  1} \rangle^2 \otimes \boldsymbol{I}_p  \right) \ \text{at} \ t = 1,
\\
\boldsymbol{R}_{2n} 
&= 
\left( \boldsymbol{I}_p +  \displaystyle \frac{\boldsymbol{C}_p}{n} \right) + \left( \displaystyle  \frac{1}{ \sqrt{n} } \langle \boldsymbol{\varphi}, \boldsymbol{u}_{\varphi  2} \rangle \otimes \boldsymbol{I}_p +  \frac{1}{ 2n } \langle \boldsymbol{\varphi}, \boldsymbol{u}_{\varphi  2} \rangle^2 \otimes \boldsymbol{I}_p  \right)\ \text{at} \ t = 2, 
\\
\vdots \ \ \ &= \ \ \ \vdots
\\
\boldsymbol{R}_{jn} 
&=
\left( \boldsymbol{I}_p +  \displaystyle \frac{\boldsymbol{C}_p}{n} \right) + \left( \displaystyle  \frac{1}{ \sqrt{n} } \langle \boldsymbol{\varphi}, \boldsymbol{u}_{\varphi  j} \rangle \otimes \boldsymbol{I}_p +  \frac{1}{ 2n } \langle \boldsymbol{\varphi}, \boldsymbol{u}_{\varphi  j} \rangle^2 \otimes \boldsymbol{I}_p  \right)\ \text{at} \ t = j, 
\end{align*}
for all $j \in \left\{ 1,..., n \right\}$. Furthermore, in this paper we focus in the case of \textit{locally explosive processes} which implies that for the diagonal matrix $\boldsymbol{C}_p = \mathsf{diag} \left( c_1,..., c_p \right)$, it holds that $c_i > 0 \ \forall \ i \in \left\{ 1,..., p \right\}$. Regardless of the additional terms in the particular representation of a local-to-unity process, it has been proved that the partial-sum process of $\boldsymbol{x}_t$ weakly convergence to the following limit process
\begin{align*}
\frac{ \boldsymbol{x}_{ \floor{nr} } }{  \sqrt{n} } \Rightarrow  \boldsymbol{G}_{ c, \varphi }(r) 
:= 
\mathsf{exp} \big\{ r \boldsymbol{C}_p + \boldsymbol{\varphi}^{\prime} \boldsymbol{B}_{ u_{ \varphi } }(r)  \otimes \boldsymbol{I}_p \big\} \left( \int_0^r  \mathsf{exp} \big\{ - s \boldsymbol{C}_p - \boldsymbol{\varphi}^{\prime} \boldsymbol{B}_{ u_{ \varphi } }(s) \big\} d\boldsymbol{B}_{ u_{x}  } (s) \right), r \in [0,1]. 
\end{align*}

\newpage

\begin{remark}
The specification of the autocorrelation coefficient given by expression \eqref{LSTUR} implies that there is a second unobserved source of variation to the local-to-unity behaviour of the nonstationary autoregressive process, which can be interpreted as exogenous variation to the system. Specifically, this approach introduces a time specific variation through this exogenous vector of covariates which can capture the effect of the reaction of the dependent variable (e.g., firm stock returns) to these unobserved (exogenous) shocks (e.g., such as the economic policy uncertainty).      
\end{remark}
  
Consider the stochastic process given by 
\begin{align}
\label{SLUR.process}
X_t = \beta_{nt} X_{t-1} + u_t, \ \ \beta_{nt} = \mathsf{exp} \left( \frac{c}{n} + \frac{ \varphi^{\prime} v_t }{ \sqrt{n} } \right) \ \ \ \text{with} \ \ \ X_0 = u_0    
\end{align}
where $c \in \mathbb{R}_{+}$ and $\varphi \in \mathbb{R}^d$. The particular specification implies that the autoregressive coefficient corresponds to a stochastic time varying parameter that fluctuates in the vicinity of unity according to the proprerties of the exogenous covariate $v_t$, the degree of persistence $c$, the sample size $n$. Based on these settings one can then estimate the unknown parameter pair $( c, \varphi )$ using a NLLS estimation approach (see, \cite{lieberman2020hybrid}) such that 
\begin{align}
\big( \widehat{c}_n, \widehat{\boldsymbol{\varphi}}_n \big) := \underset{ ( c, \varphi ) \in \Theta }{ \mathsf{arg min} } \ \sum_{t=1}^n \big( X_t - \beta_{nt} ( c, \varphi ) X_{t-1} \big)^2 
\end{align}

Overall, in this paper similar to the spirit of \cite{koul1990weakly}, we consider the construction of estimators that are asymptotically efficient over a range of nuisance parameters, known as \textit{adaptive}. Thus in our modeling environment, we consider two nuisance parameters: one being the degree of persistence, $c$, and the second being the unknown coefficient, $\boldsymbol{\varphi}_t := \big( \varphi_{1},..., \varphi_{d} \big)$, that corresponds to the exogenous covariate $\boldsymbol{v}_t$. Furthermore, we are concerned with the least squares estimators and endogenously instrumental generated estimators of the predictive and threshold predictive regression models when the nonstationary process that generate the regressors have an additional exogenous variation. In particular, our theory relies on making use of the \textit{Ornstein-Uhlenbeck processes within an uncertain environment}, since the coefficients $\left\{ \beta_{nt} \right\}_{t=1}^n$ are random processes.    

\medskip

Thus in relation to the process \eqref{SLUR.process}, we define the \textit{Ornstein-Uhlenbeck process in a random external environment} $\left\{ G_{c, \varphi}(t) \right\}_{ t \in [0,1]}$ driven by the process $G$ such that 
\begin{align*}
\frac{1}{ \sqrt{n} } X_{\floor{nt} } := \frac{1}{ \sqrt{n} }  \sum_{t=1}^{ \floor{nt} } X_t \Rightarrow  G_{c, \varphi}(t) = e^{ t c + \varphi^{\prime} B_v (t)  } \left( \int_0^t e^{ - sc - \varphi^{\prime} B_u (s) } d B_u(s) - \varphi^{\prime} \Delta_{uv} \int_0^s e^{ - sc - \varphi^{\prime} B_v(s) } ds \right).    
\end{align*}
Furthermore, consider the autoregressive regression model with stochastic local unit roots \eqref{SLUR.process} and assume that the coefficient $\varphi$ is known. Then,   \cite{lieberman2020hybrid} proved that
\begin{align}
\big( \widehat{c}_n - c \big) \Rightarrow \left( \int_0^1 G^2_{c, \varphi}(r) dr  \right)^{-1} \left( \int_0^1 G^2_{c, \varphi}(r) d B_u (r) + \Delta_{vu}^{\prime} \varphi \int_0^1 G_{c, \varphi}(r) dr + \lambda_{uu}  \right)     
\end{align}

\newpage 

Consequently, in the case when the parameter vector $\varphi = 0$, the above limit result becomes
\begin{align}
\big( \widehat{c}_n - c \big) \Rightarrow \left( \int_0^1 G^2_{c}(r) dr  \right)^{-1} \left( \int_0^1 G^2_{\varphi}(r) d B_u (r) +  \lambda_{uu}  \right)     
\end{align}
which corresponds to the standard asymptotic theory result under the LUR specification of the autocorrelation coefficient without the additional exogenous variation term. 

\begin{remark}
In the case when the parameter vector is considered to be unknown, then the consistent estimation of this nuisance parameter is more challenging and depends on whether or not we assume the presence of endogeneity in the system such as $( \Sigma_{vu}  \neq 0 )$. The main result regarding this issue proved in the study of \cite{lieberman2020hybrid} is that the distribution of $\widehat{\varphi}$ depends on the localizing coefficient $c$ through the approximation of the functional $G_{c, \varphi}(r)$. Specifically, the convergence rates in the theorem imply that the estimator is consistent when $\Sigma_{vu} = 0$. On the other hand, under the presence of endogeneity, that is, $\Sigma_{vu} \neq 0$, that parameter vector $\varphi$ may be estimated consistently using instrumental variables (see Section \ref{Section.IVX} and \cite{liu2022robust}).        
\end{remark}

Our next concern is the estimation of the threshold effects in the linear predictive regression model with locally persistence processes (i.e., or locally explosive processes in our settings). A key contribution of our study is that we focus on comparing the asymptotic behaviour of both the classical least squares estimator as well as the instrumental variable based estimator (IVX). Although the study of \cite{chen2022estimation} considers a similar setting of a threshold regression model with stochastic local unit root regressors and then develop a statistical testing procedure for the presence of threshold effects using the sup OLS-Wald test statistic, the limiting distribution is not free of any nuisance parameters due to its dependence on the unknown degree of persistence. On the other hand, in this paper we implement the IVX implementation of the sup-Wald test in the threshold predictive regression model that has been recently proposed by \cite{liu2022robust}, in the case of the linear conditional mean and conditional quantile predictive regressions models.

Notice that the proposed data generating process that allows to consider departures from unit root behaviour involve temporary departures from unity at any sample point that can move the process in stationary or explosive directions. An alternative approach than the one we focus in our work, is to consider the unknown localizing coefficient of persistence to be time-varying such that $c_n \equiv c_n (t/n)$. Such type of functional coefficient representation is examined in the studies of \cite{bykhovskaya2020point} and \cite{bykhovskaya2018boundary}. Although, in this scenario the autocorrelation coefficient varies with time and can be arbitrarily close to unity the limit process of partial sum processes converge to the standard OU process rather than the OU process in a random environment as we have in this paper. The particular limit appears regardless of the presence of the threshold variable, which indeed complicates the asymptotic theory when considering the convergence results of estimators and test statistics under the assumption of diminishing threshold effects.

\newpage 

The predictive regression model is a popular model widely used the past two decades for investigating the predictability puzzle in stock returns. Specifically, the stock return predictability literature examines various statistical and empirical aspects; seminal studies include those of \cite{pesaran1995predictability} and \cite{stambaugh1999predictive}. A particular phenomenon discussed in the more recent literature is the detection of the so-called "pockets of predictability", which means that the presence of predictability is characterized by threshold effects and nonlinearities in relation to macroeconomic events and financial stability. On the other hand, a more realistic feature is to assume that the persistence properties of predictors can be stochastic, alternating along with the aforementioned features. In this paper, we study aspects of estimation and inference in threshold predictive regression models with \textit{locally explosive regressors}. In particular, we develop the asymptotic theory for the OLS based estimator as well as the instrumental variable (IVX) based estimator proposed by \cite{magdalinos2009limit} in threshold predictive regression models.    

Various studies proposed inference methodologies for threshold models under the assumption of time series stationarity such as in \cite{hansen1996inference}, \cite{hansen2000sample}, \cite{caner2001threshold}, \cite{gonzalo2002estimation}, \cite{pitarakis2008comment}, \cite{galvao2011thresholdt, galvao2014testing}, \cite{kourtellos2014structural, kourtellos2017endogeneity} and \cite{chiou2018nonparametric}. Furthermore, in cointegrated and predictive regression models, \cite{gonzalo2006threshold, gonzalo2012regime, gonzalo2017inferring} and \cite{chen2015robust} proposed predictability tests under the presence of threshold effects for persistent regressors modelled by the local-unit-root specification. The particular predictability tests inspired by the IVX estimator of \cite{phillipsmagdal2009econometric} are examined by \cite{kostakis2015Robust} in linear predictive regressions with abstract degree of persistence. 

In addition, the time series econometrics literature also focused on developing methodologies for modelling financial bubbles which include the study of explosive and nonstationary processes as presented in the papers of \cite{phillips2007limit}, \cite{nielsen2010analysis}, \cite{magdalinos2009limit}, \cite{guo2019testing}, \cite{pedersen2020testing} and \cite{duffy2021estimation}. Furthermore, another relevant aspect is the development of econometric methods for detecting and dating market exuberance, with some relevant studies being those of \cite{phillips2011explosive} and \cite{phillips2011dating} in which the focus is both estimation of such models under the presence of explosive bubbles as well as the dating of the exact appearance of this event. Moreover, \cite{farmer2019pockets}, 
\cite{demetrescu2020testing} and \cite{georgiev2021extensions} consider the modeling and testing for the presence of time period specific predictability while \cite{bykhovskaya2020point} consider modelling a time-varying degree of persistence as a way to capture these effects. 

Therefore, motivated by these stylized facts in many financial and macroeconomic data,  \cite{lieberman2017multivariate, lieberman2020hybrid} acknowledge the shifting of persistence in time series and develop a more general specification of nonstationary processes which includes both the unknown persistence in the time series properties of regressors as well as a stochastic departure from local to unity behaviour. In this direction, the study of \cite{chen2022estimation} propose a framework for estimation and inference in threshold regression with hybrid stochastic local unit root regressors. These authors develop their own asymptotic theory results which are useful when considering estimation in a threshold environment with stochastic unit root regressors (see also \cite{liu2022robust}).

\newpage

In terms of the estimation and identification of threshold effects and nonlinearities  in predictability (e.g., \cite{maasoumi2002entropy}, \cite{gonzalo2012regime} and \cite{michaelides2016non}) via autoregressions\footnote{A seminal work on the threshold model for autoregression processes is proposed by \cite{tsay1989testing, tsay1998testing} while \cite{staiger1997instrumental} introduce the framework of IV regression with weak instruments.}  is crucial for correctly modelling economic phenomena such as price asymmetries and herd behaviour in financial markets (see, \cite{lux1995herd}, \cite{brunnermeier2009deciphering} and \cite{kapetanios2014nonlinear}) as well as in valuing macroeconomic fundamentals (see, \cite{bohn1998behavior} and \cite{hatchondo2016debt}).
When these features are not correctly captured they can manifest as structural breaks resulting to biased parameter estimates. On the other hand, the structural break literature emphasized the presence of shifts in persistence  as in \cite{horvath2020sequential}. 

Therefore, the inclusion of hybrid stochastic unit roots to the threshold predictive regression model is a novel feature which can accommodate more realistic aspects when modelling return predictability. In terms of the econometric challenges appeared within our setting these include the problem of a nuisance parameter identification, only under the alternative hypothesis. To deal with this aspect, we follow similar techniques as the ones proposed by \cite{davies1977hypothesis, davies1987hypothesis}, \cite{andrews1993tests} \cite{andrews1994optimal}, and \cite{hansen1996inference}. Our contributions in this paper are threefold. Firstly, we study relevant aspects to the estimation and inference in threshold predictive regression models with a more general persistence properties, and specifically focusing on \textit{locally explosive processes} contributing this way to the recent predictability literature. Secondly, we study the applied and theoretical implementation of the IVX estimator proposed by \cite{phillipsmagdal2009econometric} within our setting. Thirdly, an empirical application examines the regime-specific predictability hypothesis using the economic uncertainty index as the switching threshold effect. 

\paragraph{Structure.}

The rest of the paper is organized as follows. In Section \ref{Section2}, we introduce the model estimation procedure and main assumptions of the proposed econometric environment. In Section \ref{Section3}, we develop the asymptotic theory for the two proposed estimators under the presence of a threshold effect with locally explosive regressors. In Section \ref{Section4}, we examine the finite-sample performance of the sup-Wald type test statistics and summarize our main findings based on an extensive Monte Carlo simulation study. Section \ref{Section5} provides empirical evidence of threshold effects and regime-specific predictability using financial ratios of equity indices. Section \ref{Section6}, summarizes our main conclusions. The proofs of the main results in the paper can be found in the Appendix. 

\paragraph{Notation.} We denote with $\left( B_t \right)_{ t \in [0,1] }$ the Brownian motion defined on the probability space $\left( \Omega, \mathcal{F}, \mathbb{P} \right)$ which is independent of the $\sigma-$field $\mathcal{F}$. The limiting process $G$ is an $\mathcal{F}-$conditional Gaussian martinagle with $\mathcal{F}-$conditional mean. In particular, for each $t > 0$, $G_t$ or equivalently $G(t)$ has a mixed normal distribution. Moreover, we denote with $\norm{ . }_F$ and $\norm{ . }_1$ the spectral norm, $L_2$ (Frobenius) and the $L_1$ norm respectively. The symbol$=_d$ denotes equivalence in distribution and the notation $\to_p$ convergence in probability, while $\Rightarrow$ denotes weak convergence in the function space $C([0,1])$ equipped with a suitable topology.

\newpage

\section{Model Estimation and Assumptions}
\label{Section2}

Consider the following threshold predictive regression model
\begin{align}
\label{model1}
y_t = 
\begin{cases}
\alpha_1 + \beta_1^{\prime} x_{t-1} + u_{yt} , \ \ q_{t-1} \leq \gamma \\
\alpha_2 + \beta_2^{\prime} x_{t-1} + u_{yt} , \ \ q_{t-1} > \gamma
\end{cases}
\end{align}
where $y_t \in \mathbb{R}$ and $x_t$ is a $p-$dimensional vector of predictors parametrized by a vector of locally explosive processes (a special case of locally stochastic unit roots) given by expression \eqref{LSTUR}. In the special case that the model has a single regressor then its assumed to be generated by 
\begin{align}
\label{model2}
x_{t} &= \rho_{nt} x_{t-1} + u_{xt}, \ \ \rho_{nt} = \mathsf{exp} \left\{ \frac{c}{ n } +  \frac{1}{ \sqrt{n} } \langle \varphi, u_{ \varphi t}  \rangle \right\}. 
\end{align} 
where $c$ is the localizing coefficient of persistence such that $c_j > 0  \ \forall \ j \in \left\{1,...,p \right\}$ and $\upepsilon_{k,t}$ is $p \times 1$ vector with the initial condition given by $x_{1} = u_{x1}$. The particular hybrid local unit root specification allows for general types of dependence to be modelled by the threshold predictive regression given by \eqref{model1} and \eqref{model2} with $t \in \left\{ 1,...,n \right\}$ but we focus on \textit{locally explosive processes}. 

\begin{remark}
The above proposed structure of the predictive regression model implies that although the autocorrelation matrix $R_{nt}$, is not time-varying (dynamic) it is not time invariant either, which implies that is a function of time $t$. In other words, this implies that the autocorrelation coefficient has different effect throughout time, although the degree of persistence remains fixed.  In this paper, we interpret this aspect in terms of the persistence properties of regressors as a \textit{locally explosive process}. Although we do not test for the null of stationarity against the alternative of locally explosive processes, our interest is in the identification and estimation of the threshold predictive regression model with locally exposive regressors. In other words, these processes have a stochastic component which implies a stochastic explosive regime when the degree of persistence is positive.   
\end{remark}

Consider the vector $\xi_t = \left( u_{yt}, u^{\prime}_{xt}, u^{\prime}_{ \varphi t}  \right)^{\prime}$, where both both $u^{\prime}_{xt}$ and $ u^{\prime}_{ \varphi t}$ are $p-$dimensional time series vectors, such that $\xi_t$ is a strictly stationary martingale difference sequence. The partial sums process satisfy the invariance principle (\cite{phillips1987time, phillips1987towards}, \cite{phillips1988testing}) such that
\begin{align}
\label{ip}
\frac{1}{ \sqrt{n} } \sum_{t=1}^{\floor{nr} } \xi_t \Rightarrow \boldsymbol{B}_{\xi } (r) := \big[ B_{u_y } (r) , \boldsymbol{B}_{ u_x} (r), \boldsymbol{B}_{ u_\varphi }(r) \big]^{\prime} \equiv \text{BM} \left( \mathbf{\Sigma}_{\xi \xi} \right),
\end{align}
where the covariance matrix $\mathbf{\Sigma}_{\xi \xi}$ is defined as below
\begin{align}
\boldsymbol{\Sigma}_{\xi \xi} 
:= 
\begin{bmatrix}
\sigma_{y} & 0 & 0 \\
0 & \boldsymbol{\Sigma}_{xx} & 0 \\
0 & 0 & \boldsymbol{\Sigma}_{\varphi \varphi }\\
\end{bmatrix},
\end{align}
where $\mathbf{B}_{\xi }(r)$ is a vector Brownian motion with a positive-definite matrix  $\mathbf{\Sigma}_{\xi \xi}$ such that all these covariance matrices are positive-definite (for the innovation structure, see \cite{phillips1992asymptotics}). 

\newpage

The econometric identification of the threshold variable is implemented via the conditional mean specification function given by  expressions \eqref{model1}-\eqref{model2}, which implies that the presence of a threshold effect is represented 
by a fixed threshold parameter which is data-driven such that $\gamma \in \Gamma := \left[ \gamma_1, \gamma_2 \right]$. The lower threshold and upper threshold bounds are estimated such that $\mathbb{P} \left( q_t \leq \gamma_1 \right) = \pi_1 > 0$  and $\mathbb{P} \left( q_t > \gamma_1 \right) = \pi_2 < 1$. Moreover, we define the indicator variables to determine the threshold regime by $I_{1t} \equiv \mathbb{P} \left( q_t \leq \gamma \right)$ and  $I_{2t} \equiv \mathbb{P} \left( q_t > \gamma \right)$. Then, by replacing the threshold variable with a uniformly distributed random variable, $U \sim Unif[0,1]$, and using the transformation principle\footnote{In practice this property allows us to estimate the threshold variable using the uniform distribution without having to impose additional regularity conditions on the CDF of the unknown threshold.} which allows us to transform the CDF of any random variable to that of uniformly distributed random variables such that $I \left( q_t \leq \gamma \right) = I \left( F(q_t) \leq F(\gamma) \right) \equiv I \left( U_t \leq \lambda \right)$. 

For notation convenience, we express the threshold predictive regression model in a matrix form. To do this, we denote with $y$ the vector stacking $y_t$ and $X_i$ the matrix stacking $\big(  I_{it} \ x_t I_{it} \big)$ for $i = 1,2$ such that $\boldsymbol{y} = \boldsymbol{X}_1 \boldsymbol{\theta}_1 + \boldsymbol{X}_2 \boldsymbol{\theta}_2 + \boldsymbol{u}$, with $\boldsymbol{X} = \big( \boldsymbol{X}_1 \ \boldsymbol{X}_2 \big)$, $\boldsymbol{\theta} = \big( \theta_1, \theta_2  \big)$ and $\theta_i = \left( \alpha_i, \beta_i \right)$ for $i =1,2$. Alternatively, in the case when the model has only slopes, we can define $y = X(\gamma) \theta + u$ with $\theta = ( \beta^{\prime}, \delta^{\prime} )$ and $X_t(\gamma ) = \left[ x_t^{\prime} , x_t^{\prime} I \left( q_t \leq \gamma \right) \right]$. For the remaining of the paper we use the symbol $\top$ to denote the transpose of a matrix or a vector. 

Hence, given a nonstochastic threshold parameter $\gamma \in \left[ \gamma_1, \gamma_2 \right]$ we can estimate $\theta ( \gamma )$ with
\begin{align}
\widehat{ \boldsymbol{\theta} }( \gamma ) = \big[ \boldsymbol{X}(\gamma)^{\top} \boldsymbol{X}(\gamma) \big]^{-1} \boldsymbol{X}( \gamma )^{\top} \boldsymbol{y}, 
\end{align} 
Then, the least squares estimator of $\gamma$ is given by the following optimization function
\begin{align}
\widehat{\gamma} = \underset{ \gamma \in [ \gamma_1, \gamma_2  ]  }{ \text{ arg min } }  \ \text{SSR} \left( \gamma \right), 
\end{align}
where the objective function is defined as
\begin{align}
\text{SSR} \left( \gamma \right) = \sum_{t=1}^n \big( y_t -\widehat{ \boldsymbol{\theta} }(\gamma)^{\top} \boldsymbol{X}_t(\gamma) \big)^2.
\end{align}
and $\text{SSR} \left( \gamma \right)$ represents the concentrated sum of squared errors. To have consistent and efficient statistical inferences, we consider that the following general assumptions hold. In particular, for the consistent estimation of the threshold effects, we impose related regulatory conditions which are given by Assumptions \ref{assumption1}-\ref{assumption4} below.

\begin{assumption}
\label{assumption1}
$\left\{ q_t \right\}$ is assumed to be a strictly stationary and ergodic mixing sequence with a mixing coefficient $\psi_{j}$ satisfying $\sum_{j=1}^{\infty} \psi^{ \frac{1}{j} - \frac{1}{\nu} } < \infty$ for some $\nu > 2$.  
\end{assumption}

\begin{assumption}
\label{assumption2}
Let $\mathcal{F}_{n,t}$ be the smallest $\sigma-$field generated by $\left\{ \left( q_{j+1}, \boldsymbol{\xi}_j^{\top} \right): 1 \leq j \leq t \leq n  \right\}$. Then, $\left\{ \boldsymbol{\xi}_t , \mathcal{F}_{n,t} \right\}_{t=1}^n$ is a strictly \textit{martingale difference sequence} (MDS) with a positive definite covariance matrix $\mathbb{E} \left( \boldsymbol{\xi}_t \boldsymbol{\xi}_t^{\top} | \mathcal{F}_{n,t} \right)$ whose partial sums satisfy the invariance principle given by \eqref{ip} that permits the existence of asymptotic moment matrices. 
\end{assumption}

\newpage

\begin{assumption}
\label{assumption3}
We assume the existence of a threshold variable $q_t$ with continuous distribution function given by $F(.)$ and a positive continuous density function $f(.)$ such that $0 < f( \gamma ) < \bar{f} < \infty$ for all $\gamma \in \Gamma = \left[ \gamma_1, \gamma_2 \right]$, a suitable parameter space. \end{assumption}

Assumption \ref{assumption1} provides a condition for the presence of strictly stationary threshold variable $q_t$. Moreover, Assumption \ref{assumption2} states that the threshold variable $q_t$ is contemporaneously exogenous in the model and ensures that the multivariate invariance principle for the partial sum of the martingale difference array \eqref{model2} holds. Then, with Assumption \ref{assumption3} we impose the existence of a time-invariant continuous distribution function for the threshold variable, which ensures the existence of dense true threshold levels as the sample size increases. Furthermore, we can impose additional assumptions which ensure that the threshold variable induces diminishing effects which vanish asymptotically. We denote with $W(r, \lambda )$ to be a two-parameter Brownian motion on the topological space $(r,\lambda) \in [0,1]^2$ in a similar manner as in the paper of \cite{caner2001threshold}.

\medskip 

We define the following stochastic process\footnote{Notice that the particular stochastic process (i.e., notanionwise) can be seen as the solution of the stochastic differential equation: $dX_t = X_t d \widetilde{Z}_t + d Z_t$ such that 
$$
X_t = \mathsf{exp} \left\{ \widetilde{Z}_t - \frac{1}{2} \langle \widetilde{Z} \rangle_t  \right\} . 
\left( X_0 + \int_0^t \mathsf{exp} \left\{ ... \right\} \right)$$}
(see, \cite{lieberman2020hybrid})
\begin{align}
\label{Gprocess}
\frac{ \boldsymbol{x}_{ \floor{nr} } }{  \sqrt{n} } \Rightarrow  \boldsymbol{G}_{ c, \varphi }(r) 
&:= 
\mathsf{exp} \big\{ r \boldsymbol{C}_p + \boldsymbol{\varphi}^{\prime} \boldsymbol{B}_{ u_{ \varphi } }(r)  \otimes \boldsymbol{I}_p \big\} \left( \int_0^r  \mathsf{exp} \big\{ - s \boldsymbol{C}_p - \boldsymbol{\varphi}^{\prime} \boldsymbol{B}_{ u_{ \varphi } }(s) \big\} d\boldsymbol{B}_{ u_{x}  } (s) \right),
\end{align}
such that 
\begin{align}
\label{process2}
\mathbf{G}_{c , \varphi }(s) = \big[ G_{ c_1 , \varphi_1 }(s), ...., G_{ c_p, \varphi_p }(s) \big]
\end{align}

\medskip

\begin{assumption}
\label{assumption4}
The following condition holds
\begin{align}
\int_0^1 \mathbf{G}_{c , \varphi }(s) \mathbf{G}^{\prime}_{ c , \varphi }(s) ds > 0.
\end{align}
is a positive-definite stochastic matrix which holds asymptotically. 
\end{assumption}

The nonlinear stochastic process above along with Assumption \eqref{assumption4} is instrumental when developing the asymptotic theory of our test statistics. Notice that since our proposed framework allows for the existence of hybrid stochastic local unit roots the diffusion processes have a nonlinear structure to capture the additional features. When all $a_j$ and $c_j$ are zero, then we have a threshold model with unit root regressors. Therefore, our proposed framework is more general and encompasses cases such as near-unit root or integrated regressors, which can be modelled in the case of LUR specification.
Notice that the threshold indicator has the same lag as the predictor in the predictive regression model (which is one lag less than the regressand).

\newpage

\section{Asymptotic Theory and Testing Hypotheses}
\label{Section3}

All random elements are defined with a suitable probability space, denoted by $\left( \Omega, \mathcal{F}, \mathbb{P} \right)$. Throughout the paper, all limits are taken as $n \to \infty$, where $n$ is the sample size. The symbol $"\Rightarrow"$ is used to denote the weak convergence of the associated probability measures as $n \to \infty$. The symbol $\overset{d}{\to}$ denotes convergence in distribution and $\overset{\text{plim}}{\to}$ denotes convergence in probability, within the probability space (see, \cite{billingsley1968convergence}). For further details on the martingale approximation results see \cite{hall1981martingale}.     

\subsection{Asymptotic Properties}

In this section, we examine the asymptotic theory of the OLS based estimator and the corresponding IVX based estimator for the threshold predictive regression as well as for the corresponding estimator of the threshold effect\footnote{Notice that the threshold effect can be identified within a range of values say, $-1/2 < \uptau < 1/2$. For $\uptau = 1/2$ then, the nuisance parameter is at the bound of this neighbourhood and thus there is weakly identification of the threshold variable.}. The following results provide the convergence rates for the consistent identification of the true threshold effect.   

\begin{lemma}
\label{lemma1}
Under Assumptions \ref{assumption1}-\ref{assumption4}, we have that $\widehat{\gamma} - \gamma_0 = \mathcal{O}_p(1)$. 
\end{lemma}
The proof of Lemma \ref{lemma1} can be found in the Appendix. For the case of the stationary threshold model, consistency is proved by \cite{chan1993consistency}. A stronger large sample result for the identification  of the threshold effect is given by Lemma \ref{lemma2} below. 
\begin{lemma}
\label{lemma2}
Under Assumptions \ref{assumption1}-\ref{assumption4}, we have that 
\begin{align}
n^{ 2 (1- \uptau) } \left(  \widehat{\gamma} - \gamma_0 \right) \overset{ d }{ \to } \mathcal{H}_{c , \varphi } \mathcal{L},
\end{align}
where 
\begin{align}
\label{factor}
\mathcal{H}_{c , \varphi } = \frac{ \sigma_u^2 }{ \displaystyle  f(\gamma_0) \delta_0^{^{\top}} \left[ \int_0^1 \mathbf{G}_{c , \varphi }(s) \mathbf{G}^{^{\top}}_{c , \varphi }(s) ds \right] \delta_0 }
\end{align}
\end{lemma}
and $\mathcal{L} := \underset{ r \in ( - \infty, + \infty ) }{ \text{arg max} } \ \Lambda(r)$ with $\Lambda(r)$ denotes a two-sided Brownian motion  (see, the paper of \cite{khoshnevisan1996uniform} for more details) such that 
\begin{align}
\Lambda(r) = 
\begin{cases}
W_1(r)  - \frac{1}{2} r , & \text{if} \ r > 0,
\\
0, & \text{if} \ r = 0,
\\
W_2(-r) - \frac{1}{2} |r|, & \text{if} \ r < 0.
\end{cases}
\end{align}
where $W_1(r)$ and  $W_2(r)$ are two independent standard Brownian motions on $[0, \infty )$. 

\newpage 

\begin{remark}
Notice that Lemma \ref{lemma1} and Lemma \ref{lemma2} show the consistency and the limiting distribution of the threshold estimator. The convergence rate of $\widehat{\gamma}$ only depends on the diminishing rate of the threshold effect, $\uptau$, and there is no dependence to nuisance parameters such as the coefficient of persistence $c_k$ and the stochastic term, $\varphi_k$. Moreover, the limiting distribution is similar to the asymptotic result given by \cite{hansen2000sample} and \cite{chen2015robust} but in our case the scale factor is given by $\mathcal{H}_{c , \phi }$ as in expression \eqref{factor}. 
\end{remark}

\begin{corollary}
Under Assumptions \ref{assumption1}-\ref{assumption4}, and $c = 0$, we have that 
\begin{align}
n^{ 2(1 - \uptau) } \left(  \widehat{\gamma} - \gamma_0 \right) \overset{ d }{ \to } \mathcal{H}_{\varphi } \mathcal{L}, \ \ \text{with} \ \ \mathcal{H}_{\varphi } = \frac{ \sigma_u^2 }{ \displaystyle  f(\gamma_0) \delta_0^{\top} \left[ \int_0^1 \mathbf{G}_{\varphi }(s) \mathbf{G}^{\top}_{\varphi }(s) ds \right] \delta_0 }
\end{align}
where $\mathbf{G}_{\varphi}(s) = \big[ G_{\varphi_1 }(s), ...., G_{ \varphi_p }(s) \big]$ and $G_{ \varphi_k }(s) = \displaystyle e^{ \varphi_k^{\top} B_{\upepsilon_k}(s) } \left( \int_0^s  e^{ - \varphi_k^{\top} B_{\upepsilon_k}(r) } dB_{v_k} (r) \right)$.
\end{corollary}

\medskip

\begin{corollary}
Under Assumptions \ref{assumption1}-\ref{assumption4}, and $c = 0$, we have that 
\begin{align}
n^{ 2(1 - \uptau) } \left(  \widehat{\gamma} - \gamma_0 \right) \overset{ d }{ \to } \mathcal{H}_{c} \mathcal{L}, \ \ \text{with} \ \ \mathcal{H}_{c} = \frac{ \sigma_u^2 }{ \displaystyle  f(\gamma_0) \delta_0^{\top} \left[ \int_0^1 \mathbf{G}_{c}(s) \mathbf{G}^{\top}_{c}(s) ds \right] \delta_0 }
\end{align}
where $\mathbf{G}_{c}(s) = \big[ G_{ c_1 }(s), ...., G_{ c_p}(s) \big]$ and $G_{ \phi_k }(s) = \displaystyle e^{ s c_k } \left( \int_0^s  e^{ - rc } dB_{v_k} (r) \right)$.
\end{corollary}

\medskip

\subsection{IVX instrumentation}
\label{Section.IVX}

The novelty of our framework is that we consider the instrumental variable $Z_{tn}$ which is based on the IVX methodology proposed by \cite{phillipsmagdal2009econometric}. The IVX instrument is constructed as below 
\begin{align}
z_{tn} = \sum_{j=0}^{t-1} \left( 1 - \frac{c_z}{n^{\upgamma_z}}  \right)^j \big( x_{t-j} - x_{t-j-1} \big), \ \ \ \ \text{with} \ c_z > 0 , \ \upgamma_z \in (0,1).
\end{align}  
The philosophy of the IVX instrumentation is that it induces a mildly integrated regressor which can control the unknown degree of persistence in the original regressor. In particular, in the case of hybrid stochastic local unit roots by allowing $\upgamma_z \neq 1$ which cover cases such as $\upgamma_z > 1$, $\upgamma_z \in (0,1)$ or $\upgamma_z < 1$, the IVX filter can achieve this property which is found to ensure weakly convergence to a mixed Gaussian distribution of the IVX estimator (see, \cite{kostakis2015Robust}). It remains to verify that this asymptotic property still holds in our setting. Then, this will allow us to derive standard Brownian functionals for the asymptotic results of the tests of the next section even under the presence of threshold effects and nonlinear stochastic terms. The convenience of deriving an analytical known form for the corresponding limiting distributions of the tests is that we can easily obtain critical values and conduct statistical inference. 

\newpage  

Next, we shall consider the IVX instrumentation within our modeling environment. Specifically, in the settings of SLUR IVX instrumentation has to be applied in a different manner than in the classical nonstationary autoregressive model, due the fact that the specification of the autocorrelation coefficient has a more complicated form. This is expressed in the following form
\begin{align}
\widetilde{z}_{t-1, t^{\prime}} = \sum_{j=1}^{t-1} R_{nz,t^{\prime}}^{t-j-1}  \left( u_{xj} + \frac{ \boldsymbol{C}_p }{n } +  \frac{1}{ \sqrt{n} } \langle \boldsymbol{\varphi}, \boldsymbol{u}_{\varphi  1} \rangle \otimes \boldsymbol{I}_p x_{j-1}  +  \frac{1}{ 2n } \langle \boldsymbol{\varphi}, \boldsymbol{u}_{\varphi  1} \rangle^2 \otimes \boldsymbol{I}_p x_{j-1}  \right)  
\end{align}

In practise the above instrumentation procedure is applied and we simplify it to the following terms
\begin{align}
\widetilde{z}_{t-1} 
= 
z_{t-1} + \frac{\boldsymbol{C}_p }{n} \eta_{n,t-1}^{(1)} + \frac{1}{ \sqrt{n} } \eta_{n,t-1}^{(2)} + \frac{1}{2n} \eta_{n,t-1}^{(3)} 
\end{align}
where the three terms above are defined as below
\begin{align}
\eta_{n,t-1}^{(1)} &= \sum_{j=1}^{t-1} \boldsymbol{R}_{nz}^{t-j-1} \boldsymbol{x}_{j-1},
\\
\eta_{n,t-1}^{(2)} &= \sum_{j=1}^{t-1} \boldsymbol{R}_{nz}^{t-j-1} \left( \boldsymbol{\varphi}^{\prime} \boldsymbol{u}_{aj} \right) \boldsymbol{x}_{j-1},
\\
\eta_{n,t-1}^{(3)} &= \sum_{j=1}^{t-1} \boldsymbol{R}_{nz}^{t-j-1}  \left( \boldsymbol{\varphi}^{\prime} \boldsymbol{u}_{aj} \right)^2 \boldsymbol{x}_{j-1}.
\end{align}

\begin{example}
Consider the matrix format of the second term above (with $(p = d)$)
\begin{align*}
\eta_{n,t-1}^{(2)} 
&= 
\sum_{j=1}^{t-1} \boldsymbol{R}_{nz}^{t-j-1} \left( \boldsymbol{\varphi}^{\prime} \boldsymbol{u}_{ \varphi j} \right) \boldsymbol{x}_{j-1}    
\\
&= 
\sum_{j=1}^{t-1} \boldsymbol{R}_{nz}^{t-j-1} \big( \varphi_1,...,  \varphi_d \big) \big( u_{ \varphi 1, t},..., u_{ \varphi d, t} \big) 
\end{align*}    
\end{example}
Another example, is to consider the case in which $p = d = 1$ (single SLUR regressor in the model). Then, it holds that
\begin{align}
\boldsymbol{\rho}_{nt} = \left( 1 + \frac{c}{n} \right)  + \frac{1}{\sqrt{n}} \varphi_1 \boldsymbol{u}_{at} + \frac{1}{2n} \left( \varphi_1 \boldsymbol{u}_{at} \right)^2.   
\end{align}

Notice that the above representation doesn't imply that the autocorrelation coefficient is estimated dynamically within a rolling window but that its value when the DGP is constructed is estimated at each time series observation of the full sample.  Furthermore, an IVX instrumentation is also applied to the error term $\boldsymbol{u}_{\varphi t}$ such that 
\begin{align}
Z_{n \varphi, t} := \frac{1}{ n^{ \frac{\upgamma_z}{2} } } \sum_{j=1}^t \boldsymbol{\varphi}^{\prime} \boldsymbol{u}_{\varphi j} \rho_{nz}^{t-j} \Rightarrow Z_{\varphi} =_d \mathcal{MN} \left( 0, \frac{ \boldsymbol{\varphi}^{\prime} \boldsymbol{\Omega}_{ \varphi \varphi } \boldsymbol{\varphi} }{ - 2c_z  } \right). 
\end{align}

\newpage  

\begin{remark}
Notice that the above result holds due to the R-mixing condition. Specifically, the R-mixing condition is employed to establish weak convergence of these partial sum functionals in the $D[0,1]$ space equipped with the $J_1$ topology. In particular, the R-mixing property of the weak convergence means that the random element $Z_{n\varphi, t}$ is asymptotically independent of all events $E \in \mathcal{F}$, as $n \to \infty$ which is expressed as below
\begin{align}
\mathbb{P} \big[  \big( Z_{n \varphi, t} \in . \big) \cap E \big] \to \mathbb{P} \big[  \big( Z_{n \varphi, t} \in . \big) \big] \mathbb{P} \left[ E \right].  
\end{align}
\end{remark}

\bigskip

Then, the main result regarding the IVX estimator implies that a Mixed Gaussianity assumption holds in large samples such that
\begin{align}
n^{ \frac{ 1 + \upgamma_z }{2} } \left( \widehat{\boldsymbol{\beta}}_1^{ivx} - \boldsymbol{\beta}_1 \right) \sim_{a}  \mathcal{MN} \left( \boldsymbol{0}_{ \textcolor{red}{n \times 1 } }, \mathbb{V}^{ivx} \right) 
\end{align}
where $\mathbb{V}^{ivx} \in \mathbb{R}^{ p \times p }$.  

\begin{proof}
To prove the above result consider the following analytical expression. 

Define with $\boldsymbol{R}_n  := \left( \boldsymbol{I}_p + \frac{ \boldsymbol{C} }{n} \right)$ and let the vector of regressors being expressed as below
\begin{align}
x_t = \big( \boldsymbol{R}_n x_{t-1} + u_{xt} \big) + \frac{1}{ \sqrt{n} } \varphi u_{\varphi t} x_{t-1} + \frac{1}{2n} \left( \varphi u_{\varphi t} \right)^2 x_{t-1}.    
\end{align}
Then, it can be proved that the following result holds
\begin{align}
\frac{1}{ n^{ 1 + \upgamma_z } } \sum_{t=1}^n \widetilde{z}_{t-1} \underline{x}_{t-1}^{\prime} = O_p(1).    
\end{align}
Specifically for LSTUR the following limit result holds
\begin{align*}
\frac{1}{ n^{ 1 + \upgamma_z } } \sum_{t=1}^n x_{t-1} \widetilde{z}_{t-1}^{\prime} \to_d \mathbb{V}_{xz} := \frac{-1}{c_z} \left( \int_0^1 G_{c, \varphi} (r) dB_x^{\prime} (r) + \Omega_{xx}  + \int_0^1 G_{c, \varphi} (r) \left( \varphi^{\prime}  \Omega \varphi \right) dr \right)      
\end{align*}
Next consider the second term of the $x_t$ expression, we obtain that
\begin{align*}
\frac{1}{ n^{ \frac{1 + \upgamma_z }{2} } } \eta_{n,t-1}^{(2)} 
&= 
\left( \frac{1}{ n^{ \frac{\upgamma_z }{2} } } \sum_{j=1}^{t-1}  \left( \boldsymbol{\varphi}^{\prime} \boldsymbol{u}_{aj} \right)\boldsymbol{R}_{nz}^{t-j-1} \right) \left( \frac{1}{ \sqrt{n} }  \boldsymbol{x}_{j-1} \right)  + \mathcal{O}_p \left( \frac{1}{ n^{ \frac{1 - \upgamma_z }{2} } }  \right)
\\
&=
Z_{n \varphi, t- 1} \left( \frac{1}{ \sqrt{n} }  \boldsymbol{x}_{j-1} \right)  + \mathcal{O}_p \left( \frac{1}{ n^{ \frac{1 - \upgamma_z }{2} } }  \right).
\end{align*}
(incomplete proof)
In other words, this specification implies that there is an additional effect in each regressor included in the system. Moreover, we need to determine the order of convergence for all the three terms. 

\end{proof}

\newpage  

\subsection{Testing for Nonlinearity and Predictability}

In this Section, we focus on developing two tests: (i) testing for linearity and (ii) testing for Joint Nonlinearity and Predictability based on the threshold predictive regression model described above. Therefore, we can reformulate the threshold model as below
\begin{align}
\label{form}
y = \alpha + \beta x + X_2 \eta + u
\end{align}
where $\alpha = \alpha_1$, $\beta = \beta_1$ and $\eta = \left( \alpha^{*}, \delta^{*} \right)^{\top}$ with $\alpha^{*} = \alpha_2 - \alpha_1$ and $\delta^{*} = \beta_1 - \beta_2$.  Then, we can observe that the threshold effect diminishes when $\eta = 0$. Therefore, under the null hypothesis for a fixed $\gamma \in \Gamma = \left[ \gamma_1, \gamma_2 \right]$ the Wald statistic has the following form 
\begin{align}
\mathcal{W}_n \left( \gamma \right) = \hat{\eta}^{\top} \big[ X_{\gamma}^{\top} \left( I_n - P_x \right) X_{\gamma} \big] \hat{\eta} \big/ \hat{\sigma}^2_u. 
\end{align}
where $P_x$ the projection matrix of $x_t$ and $I_n$ an $n \times n$ identity matrix. 

We define the corresponding supremum functional as in \cite{caner2001threshold}, \cite{pitarakis2008comment} and \cite{gonzalo2012regime, gonzalo2017inferring} such that 
\begin{align}
\label{test1}
\mathcal{W}_n^{*} := \underset{ \gamma \in \left[ \gamma_1, \gamma_2 \right]   }{ \text{sup} } \mathcal{W}_n \left( \gamma \right). 
\end{align}
We aim to derive the asymptotic distribution of the test statistic given by \eqref{test1}, under the null hypothesis $\mathbb{H}_0^{(1)} : \eta = 0$ and show that in the special case when $c > 0$ and $\phi = 0$, then the limit theory reduced to the asymptotic result proved by Proposition 1 of \cite{gonzalo2012regime}. Due to the presence of nuisance parameter identified only under the alternative hypothesis, we follow \cite{davies1977hypothesis} and \cite{hansen1996inference}. 

Therefore, in order to test for both non-linearity and predictability we employ the supremum functional with the unknown threshold variable being within the range of the values $\gamma_1$ and $\gamma_2$. Therefore, the estimation procedure involves scanning within this window and applying the maximizing to obtain the unknown threshold variable. When we obtain an estimate for the threshold variable, then the estimation procedure requires to substitute this estimate and then estimate the model parameters\footnote{A detailed description of the procedure can be found in the book of \cite{terasvirta2010modelling}. Moreover, \cite{tong1980thresholdAR}}.

Furthermore, related to the assumptions of the model as also argued in \cite{gonzalo2012regime} the dependence assumptions such as for example how the correlation structure between the threshold variable and the innovations of the predictive regression can affect the asymptotic theory of the predictability tests. For instance, imposing the assumption that there is no correlation\footnote{In particular, assuming certain correlation structure between the threshold variable and the innovations of the predictive regression model implies that predictability is conditioned on whether in that certain period of time, there is a high correlation between these two random quantities. } between these two quantities, provides an equivalent assumption of having an exogenous determination of the threshold variable.

\newpage 

\begin{theorem}
\label{theorem1}
Under the null hypothesis $\mathbb{H}^{(1)}_0: \eta = 0$  and Assumptions \ref{assumption1}-\ref{assumption4}, then 
\begin{align*}
\mathcal{W}_{n,OLS}^{*(1)} \Rightarrow \underset{ \gamma \in \left[ \gamma_1, \gamma_2 \right]   }{ \text{sup} } &\left\{ \int_0^1 \mathbf{G}_{c , \phi }(s)  dW \left( s, F( \gamma) \right) - F( \gamma) \int_0^1 \mathbf{G}_{c , \phi }(s) dW(s) \right\}^{\top} 
\\
&\times \bigg\{ \big[ F(\gamma) \big( 1 - F(\gamma) \big) \big] \int_0^1 \mathbf{G}_{c , \phi }(s) \mathbf{G}_{c , \phi }(s)^{\prime} ds \bigg\}^{-1} 
\\
&\left\{ \int_0^1 \mathbf{G}_{c , \phi }(s)  dW \left( s, F( \gamma) \right) - F( \gamma) \int_0^1 \mathbf{G}_{c , \phi }(s) dW(s) \right\} \big/ \sigma^2_u
\end{align*}
\end{theorem}

\medskip

\begin{remark}
Notice that Theorem \ref{theorem1} gives the asymptotic distribution under the null hypothesis that $\mathbb{H}_0^{(1)}: \eta = 0$ which depends on a set of nuisance parameters such as the unknown threshold variable as well as the coefficient of persistence. This non-pivotal property of the limit theory for the Wald-OLS statistic can complicated inference especially in obtaining critical values and testing the hypothesis of interest. However, in the special case when $\phi = 0$ and $c > 0$ then, the result reduces to the asymptotic result given by Proposition 1 of \cite{gonzalo2012regime}. 
\end{remark}

Notice that testing for $\beta_1 = \beta_2 = 0$ for a given $\lambda \in (0,1)$ then induces the Wald statistic with the estimated threshold parameter $\hat{\lambda}$ with a limiting distribution given by 
\begin{align}
W_n ( \hat{\lambda} ) \Rightarrow \chi^2(2) + \frac{ \displaystyle \left[ \int_0^1 J^{\mu}_c (r) dB_u(r,1) \right]^2 }{  \displaystyle \sigma_u^2 \int J_c^{\mu}(r)^2 }    
\end{align}
regardless of whether $\alpha_1 = \alpha_2$ or $\alpha_1 \neq \alpha_2$. 

In other words, the Wald statistic above is useful for conducting inferences regarding the presence of regime specific slopes in the predictive regression model without prior knowledge on whether the model intercepts are regime dependent or not. In other words, the limiting distribution of the sup-Wald statistic evaluated at the estimated threshold parameter, $\hat{\lambda}$, is the same regardless of whether $\alpha_1 = \alpha_2$ or $\alpha_1 \neq \alpha_2$. Moreover, due to the presence of the nuisance parameter of persistence as well as the presence of endogeneity in the system, such that, the allowed correlation between the Brownian motions $B_u$ and $B_v$, then the analytical expression of its limit process has its second component depending on $\sigma_{uv}$. On the other hand, the limiting distribution above simplifies further under the requirement that $\sigma_{un} = 0$ (exogeneity assumption), such that $W_n ( \hat{\lambda} ) \Rightarrow \chi^2(2)$, as $n \to \infty$.   

Next, we focus on the null hypothesis that jointly tests the absence of linearity and no predictive power of the threshold predictive regression model which implies that  $\mathbb{H}_0^{(2)}: \alpha_1 = \alpha_2 , \beta_1 = \beta_2 = 0$. Equivalently, based on the formulation given by expression \eqref{form} the null hypothesis becomes $\mathbb{H}_0^{(2)}: \eta = 0 , \beta = 0$ and the limiting distribution\footnote{A related proof for the instrumental variable case is presented by \cite{lieberman2018iv}.} of the corresponding sup Wald-OLS and sup Wald-IVX tests are given by Theorem \ref{theorem2} below.

\newpage 

\begin{theorem}
\label{theorem2}
Under the null hypothesis $\mathbb{H}_0^{(2)}: \eta = 0, \beta = 0$  and Assumptions \ref{assumption1}-\ref{assumption4}, then 
\begin{align*}
\mathcal{W}_{n,OLS}^{*(2)} ( \lambda )  &\Rightarrow  \frac{ \displaystyle \left[ \int_0^1 G_{c , \phi }(s) dW(s) \right]^2}{ \displaystyle \sigma_u^2  \int_0^1 G^2_{c , \phi }(s) ds  } + \mathcal{W}_{n,OLS}^{*(1)} ( \lambda )
\\
\\
\mathcal{W}_{n,IVX}^{*(2)} ( \lambda ) &\Rightarrow W(1)^2 + \underset{ \lambda }{ \text{sup} } \ \frac{\mathcal{BB}(\lambda)^{\top}\mathcal{BB}(\lambda)}{ \lambda(1 - \lambda)},
\end{align*}
where $\mathcal{BB}(\lambda)$ the standard Brownian Bridge. 
\end{theorem}

\begin{remark}
Notice that Theorem \ref{theorem2} gives the asymptotic distribution of the sup Wald test under the null hypothesis $\mathbb{H}_0^{(2)}: \eta = 0, \beta = 0$ for both the OLS and IVX estimators. In the case we employ the IVX estimator, we verify that the limiting distribution is free of nuisance parameters and follow a similar form as the asymptotic result presented by Proposition 3 of \cite{gonzalo2012regime}. The significance of our result  is that we generalize the robust property of the IVX estimator in the case of predictors with hybrid stochastic local unit roots.
\end{remark}

\begin{remark}
 For the proof of Theorem 2 we focus on the IVX instrumentation in the case of stochastic local unit roots regressors in the predictive regression model. The main idea of this theorem is to prove that the IVX-based estimation methodology results in a limiting distribution of the sup-Wald statistic which is equivalent to the one obtained under strict exogeneity.  Similarly, for Theorem 2 above, under the assumption of an estimated threshold variable, $\hat{\lambda}$, then the limiting distribution of the  $W_{n, ivx } ( \hat{\lambda} ) \Rightarrow \chi^2(2)$, as $n \to \infty$ since the supremum functional is replaced by the exact estimated value of the threshold parameter.  
\end{remark}


\section{Simulation Experiments}
\label{Section4}

In this Section, we examine the finite sample performance of the proposed estimators and predictability tests via an extensive Monte Carlo simulation study.

\subsection{Data Generating Process}

We consider the following data generating process  
\begin{align}
y_t &= \alpha_0 + \frac{2}{n^{0.25} } \boldsymbol{x}_{t-1} I \left( q_t \leq \gamma_0 \right) + u_{t}, 
\\
\begin{bmatrix}
x_{1t}
\\
x_{2t}
\end{bmatrix}
&= 
\begin{bmatrix}
\rho_{1} & 0 
\\
0 & \rho_{2}
\end{bmatrix}
\begin{bmatrix}
x_{1t-1}
\\
x_{2t-1}
\end{bmatrix} + v_{t} \\
\uprho_{nj} &= \text{exp} \left\{ \frac{c}{ n } + \frac{ \alpha \epsilon_{t} }{\sqrt{n}} \right\}, 
\end{align}
where $q_t$ and $\xi_t = \left( u_t, v_t, \upepsilon_t \right)$ are independently normally distributed with mean zero and variance one and $\gamma_0 = 0.25$. Furthermore, we need to choose appropriate values for the unknown parameter $\varphi$ and the nuisance parameter of persistence. We consider that the coefficient of persistence takes values such that $c \in \left\{ 1, 2, 5, 10 \right\}$ and $\phi \in \left\{ 0, 0.05 , 0.25, 0.50 \right\}$. Each data generating process is replicated based on $B = 5,000$ and we consider sample sizes such as $n \in \left\{ 250, 500 \right\}$.

We obtain the threshold estimator based on both the OLS and the IVX estimation and compare the empirical size results. In summary we observe that the performance of the threshold estimator improves as the sample size increases. We also observe the diminishing threshold effect for the threshold estimator across different values of $(c , \phi)$ and $n$. However, when we observe the threshold estimators based on the OLS versus the IVX procedure in relation to the coefficient of persistence, we can clearly see that the IVX estimator produces smoother convergence rates and empirical sizes closer to the nominal size especially under the assumption of high persistence in the regressors. 

\

\subsection{Empirical size and power of the predictability tests}

Under the null hypothesis of no threshold effects, the model reduces to the standard predictive regression model, which implies that there is absence of nonlinearity. Thus, the empirical size is obtained with the use of the sup-Wald statistics by replicating the DGP and counting the frequency of rejecting the null hypothesis with respect to the replications. Similarly, we can obtain the empirical power of the tests under the alternative hypothesis of nonlinearity and predictability.  




For the experimental design we consider $B = 1000$ and $n = \left\{ 250, 500 \right\}$ to simplify the computational time. Moreover, we consider a predictive regression model with three types of regressors, such that \textit{(i)} mildly integrated regressors (where the degree of persistence implies that the LUR component is close to the unit boundary), \textit{(ii)} mildly explosive regressors (where the degree of persistence implies that the LUR component is on the explosive side of the unit boundary) and  \textit{(iii)} near nonstationary regressors (where the degree of persistence implies that the LUR component is well below the unit boundary). Furthermore, to access the performance of the proposed test statistics with respect to the relative efficiency of the IVX to the OLS estimator, we only consider the case in which all the regressors are locally explosive.

In addition, to exclude other effects, we assume that the degree of endogeneity in the system remains the same across all variables and is given by the a fixed covariance matrix, which is unchanged throughout the simulation step. Then, under the null hypothesis of no threshold effect both test statistics are constructed under the assumption that the parameter vector across the two regimes remains the same, such that, $\beta_1 = \beta_2$ (e.g., in the case when we assume that the model includes no intercepts) and $\theta_1 = \theta_2$, where $\theta_j = (  \alpha_j, \beta_j )$ for $j \in \left\{ 1,2 \right\}$ when the model includes both a slope and an intercept. Furthermore, we use a pre-specified threshold cut-off point which also remains fixed through the replication steps. Another used specific parameter which is the coefficient of persistence for the instrumentation procedure, this is also kept fixed to avoid additional complexitiy when comparing the finite sample performance of the test statistics under the null hypothesis. 


\newpage 


\section{Empirical Application}
\label{Section5}

In this Section, we examine an empirical implementation of the proposed framework. Our goal is to uncover regime-specific predictability and threshold effects in financial markers, focusing on the pre-pandemic and post-pandemic sampling periods. More specifically, our aim is to assess whether the data support the presence of regime-specific predictability due to the socio-economic events resulted from the 2019 pandemic. 

\subsection{Regime-Specific Predictability}

The stock return predictability is a major puzzle in financial economics. The literature goes back several decades; we briefly highlight important studies. The seminal paper of \cite{perron1989great} discuss the unit root hypothesis around periods of financial turbulence. \cite{dejong1991temporal} study the debate whether divided are trend-stationary or integrated processes while \cite{stein1991stock} propose a framework for stock price distributions and stochastic volatility (see, also \cite{menzly2004understanding}). Moreover, \cite{timmermann1993learning} examines the excess volatility and predictability puzzle in stock prices. Further studies related to predictability testing include \cite{campbell2006efficient}, \cite{kostakis2015Robust} \cite{kasparis2015nonparametric}, \cite{demetrescu2020testing}. In this paper, we focus on the regime-specific\footnote{Notice that for instance formal econometric methodologies for testing for regime switching are proposed by \cite{cho2007testing}, however in this paper our focus is to examine whether we find statistical significant evidence of regime-specific predictability rather than to test for regime-switching.} predictability hypothesis as proposed by the studies of \cite{gonzalo2012regime, gonzalo2017inferring}. 

We are also inspired by the work of \cite{hatchondo2016debt}\footnote{The particular framework even though is driven from the perspective of sovereign debt it includes the main idea that an endogenous threshold variable such as the price at which long-term bonds would trade without current-period borrowing, drives the identification strategy of the model.} who introduce the idea of a threshold macroeconomic variable driving economic policy making as well as the paper of \cite{atanasov2020consumption} who examine consumption fluctuations and predictability of expected returns. Therefore, we are motivated to study an additional aspect not previously examined in the predictability literature such as the effect of economic policy uncertainty as a potential threshold variable. \cite{baker2016measuring} introduce a measure for economic policy uncertainty and examine the level of predictable policy responses\footnote{A different aspect but of possible interest is how systemic risk can affect stock return predictability. For instance, the aspects of uncertainty and systemic risk are examined by \cite{dicks2019uncertainty}. Moreover, a discussion on the aspects of fluctuations in uncertainty is provided by \cite{bloom2014fluctuations}.}.  

In particular, economic policy uncertainty can be considered as a set of shocks which affects the predictability of returns in relation to the commonly used predictors in the literature. Therefore, we aim to use the indicator of economic policy uncertainty introduced by Baker et al. (2016) as an exogenous threshold variable to assess the existence of regime-specific predictability.

\newpage 



\section{Discussion}
\label{Section6}

In this paper, we study a special special class of persistence which is the class of hybrid stochastic local unit roots, for the threshold predictive regression model, filling the gap in the literature of predictability tests. Our study extends the current methodologies of identifying regime specific predictability with persistence or unit root predictors. In particular, by incorporating such persistence shifts with the framework proposed by \cite{lieberman2020hybrid} it allows us to reformulate the predictability tests proposed by \cite{gonzalo2012regime, gonzalo2017inferring} and \cite{kostakis2015Robust} generalizing this way the particular asymptotic results. We hope that our proposed framework will be helpful to practitioners who are interested in detecting predictability and threshold effects under the presence of potential hybrid stochastic local unit roots in equity indices and financial variables capturing economic conditions and market sentiment.   

The asymptotic theory of this paper confirms that the estimation and inference of the threshold predictive regression model produces consistent parameter estimates. Moreover, the simulation experiment demonstrate good empirical size and power properties across different values of the unknown degree of persistence. The additional persistence properties we consider in this paper, allows to capture how the effect of shifting persistence as modelled via the hybrid stochastic local unit root affects the presence of the threshold effect. In terms of the IVX instrumentation procedure, the IVX instrument can filter out the abstract persistence occurred from the unknown coefficient of persistence as well as the component which comes from the additional term in the LUR specification. The additional component is considered as an extra source of innovation appeared in the system by a set of exogenous covariates. Thus, we can consider these being for example, the effect of economic shocks, such as the economic policy uncertainty to the persistence properties of regressors and the predictability in the model. Notice that a key assumption for the development of the asymptotic theory is the assumption of diminishing threshold effects. This allows us to obtain the limiting distribution of the threshold estimator as well as the asymptotic behaviour of the sup Wald statistic for testing the existence of the threshold effect under the null hypothesis. 

Finally, some important extensions to the present work are worth mentioning. Firstly, time-varying predictability can be modelled in parallel to the hybrid stochastic local to unity specification. Secondly, one can consider extending our proposed threshold predictive regression model to allow for multiple regimes and thus the presence of multiple threshold effects (see, for example \cite{chiou2018nonparametric}). This can increase significantly the complexity of the framework but nevertheless methods such as the ones proposed by \cite{gonzalo2002estimation} and \cite{gonzalo2005subsampling} could be utilized. Thirdly, considering the presence of an endogenous threshold variable or smooth transitions between regimes such as in \cite{luukkonen1988testing}. Last but not least, examining the implementation of our proposed tests in forecasting and predictive accuracy environments such as in the studies of \cite{emilianouncovering} and \cite{pitarakis2020novel} is another possibility; all these being interesting applications for future research.

\newpage 
   
\bibliographystyle{apalike}
\bibliography{myreferences1}

\newpage 

\appendix

\counterwithin{lemma}{section}

\section{Mathematical Appendix}

\begin{lemma}
\label{GP}
\citep{gonzalo2012regime} 
Under Assumption 1 and 2 as $n \to \infty$, 
\begin{itemize}
\item[(i)] $\displaystyle \frac{1}{n} \sum_{t=1}^n I_{1t} \overset{ p }{ \to } \lambda$, \ \ (ii) $\displaystyle \frac{1}{n^{3/2}} \sum_{t=1}^{ n} x_{t-1} \Rightarrow \int_0^1 K_c(r) dr$, \ \ (iii) $\displaystyle \frac{1}{n^{2}} \sum_{t=1}^{ n} x^2_{t-1} \Rightarrow \int_0^1 K^2_c(r) dr $,

\item[(iv)] $\displaystyle \frac{1}{ n } \sum_{t=1}^{ n} x_{t-1} v_{t} \overset{ p }{ \to } \int_0^1 K_c(r) dB_v (r) + \omega_{vv}$, \ \ (v) $\displaystyle \frac{1}{n} \sum_{t=1}^n x_{t-1} u_t \Rightarrow \int_0^1 K_c(r) dB_u(r,1)$, 

\item[(vi)]  $\displaystyle \frac{1}{ n^2 } \sum_{t=1}^{ n} x^2_{t-1} I_{1t} \Rightarrow \lambda \int_0^1 K^2_c(r) dr$, \ \ (vii) $\displaystyle \frac{1}{ n^{3 / 2} } \sum_{t=1}^{ n} x_{t-1} I_{1t} \Rightarrow \lambda \int_0^1 K_c(r) dr$

\item[(viii)]  $\displaystyle \frac{1}{\sqrt{n}}  \sum_{t=1}^{ \floor{nr} } u_t I_{1t} \Rightarrow B_u( r, \lambda)$, \ \ (ix)  $\displaystyle \frac{1}{n} \sum_{t=1}^{n} x_{t-1} I_{1t-1} \Rightarrow \int_0^1 K_c(r) dB_u(r, \lambda)$. 
\end{itemize}
\end{lemma}

\medskip

Let $x_{t, \gamma} = x_t I \left( q_t \leq \gamma \right)$ and $X_\gamma$ stacks up $x_{t, \gamma}$, then the following results hold. 

\begin{lemma}
\label{LP}
\citep{lieberman2020hybrid} Under Assumption 2, we have that 
\begin{align}
n^{- 1/ 2} x_{\floor{ns}} \Rightarrow \mathbf{G}_{c , \phi }(s) = \text{exp} \left\{ s c + \phi^{\top} B_{\upepsilon}(s) \right\} \left( \int_0^s  \text{exp} \big\{ - r c - \phi^{\top} B_{\upepsilon}(r) \big\} dB_{v} (r) \right)
\end{align}
\end{lemma} 

\begin{proof}
The proof of Lemma \ref{LP} is given by Lemma 1 in \cite{lieberman2020hybrid} with no presence of endogeneity in $\epsilon_t$.   
\end{proof}

\medskip

\begin{lemma}
\label{results}
Under Assumptions \ref{assumption1}-\ref{assumption4}, uniformly for any $\gamma \in \left[ \gamma_1, \gamma_2 \right]$, the following weakly converges result hold
\begin{align}
(i)& \ \ n^{-3 / 2} \sum_{t=1}^n x_t \Rightarrow \int_0^1 \mathbf{G}_{c , \phi }(s) ds,
\\
(ii)& \ \ n^{-3 / 2} \sum_{t=1}^n x_{t, \gamma} \Rightarrow F( \gamma ) \int_0^1 \mathbf{G}_{c , \phi }(s) ds  ,
\\
(iii)& \ \ n^{-1 / 2} \sum_{t=1}^{ \floor{ns} } I \left( q_t \leq \gamma \right) u_t \Rightarrow \sigma_u W \left( s, F( \gamma ) \right),
\\
(iv)& \ \ n^{-1} \sum_{t=1}^n x_{t, \gamma} u_t \Rightarrow \sigma_u  \int_0^1  \mathbf{G}_{c , \phi }(s) dW \left( s, F \left( \gamma \right) \right),
\\
(v)& \ \ n^{-2} \sum_{t=1}^n x_{t} x_{t}^{\top} \Rightarrow  \int_0^1  \mathbf{G}_{c , \phi }(s) \mathbf{G}^{\top}_{c , \phi }(s) ds,
\\
(vi)& \ \ n^{-2} \sum_{t=1}^n x_{t} x_{t}^{\top} u_t \Rightarrow F( \gamma )  \int_0^1  \mathbf{G}_{c , \phi }(s) \mathbf{G}^{\top}_{c , \phi }(s) ds,
\\
(vii)& \ \ n^{-2} \sum_{t=1}^n X_{t} \left( \gamma \right) X_{t}^{\top} \left( \gamma \right) \Rightarrow \mathcal{M}( \gamma ),
\end{align}
\end{lemma}

\newpage 

where
\begin{align}
\mathcal{M} ( \gamma ) 
= 
\begin{bmatrix}
\displaystyle \int_0^1 \mathbf{G}_{c , \phi }(s) \mathbf{G}^{\top}_{c , \phi }(s) ds \ \ &  \ \  \displaystyle F( \gamma ) \int_0^1 \mathbf{G}_{c , \phi }(s) \mathbf{G}^{\top}_{c , \phi }(s) ds
\\
\\
\displaystyle F( \gamma ) \int_0^1 \mathbf{G}_{c , \phi }(s) \mathbf{G}^{\top}_{c , \phi }(s) ds \ \ &  \ \  \displaystyle F( \gamma ) \int_0^1 \mathbf{G}_{c , \phi }(s) \mathbf{G}^{\top}_{c , \phi }(s) ds
\end{bmatrix}
\end{align}
\begin{proof}
For (i) above, this is a standard result. For the remaining results we can define $x_{nt} = n^{-1/2} x_t$ and then we can use the uniform convergence of the threshold parameter such that for $\gamma \in \left[ \gamma_1, \gamma_2  \right]$ we have that
\begin{align}
\underset{ \gamma \in \left[ \gamma_1, \gamma_2  \right] }{ \text{sup} } \ \left| \frac{1}{n} \sum_{t=1}^n x_{nt} \big( I \left( q_t \leq \gamma \right) - F(\gamma ) \big) \right| \overset{  p }{ \to } 0. 
\end{align}
\end{proof}

\begin{lemma}
Under Assumptions \ref{assumption1}-\ref{assumption4}, for any $\gamma \in \left[ \gamma_1, \gamma_2 \right]$, we have that 
\begin{align}
n^{\uptau} \left( \widehat{\theta} ( \gamma ) - \theta \right) \Rightarrow \mathcal{M} \left( \gamma \right)^{-1} \mathcal{A} \left( \gamma, \gamma_0, \delta_0 \right), 
\end{align}
\begin{align}
\mathcal{A} \left( \gamma, \gamma_0, \delta_0 \right) 
=
\begin{bmatrix}
\displaystyle \big( F \left( \gamma_0 - F(\gamma) \right) \big) \int_0^1 \mathbf{G}_{c , \phi }(s) \mathbf{G}^{\top}_{c , \phi }(s) ds
\\
\\
\displaystyle  \big( F \left( \gamma_0 \wedge \gamma \right) - F( \gamma) \big) \int_0^1 \mathbf{G}_{c , \phi }(s) \mathbf{G}^{\top}_{c , \phi }(s) ds
\end{bmatrix} \delta_0.
\end{align}
\end{lemma}

\begin{proof}
By definition we have that $\widehat{\theta} \left( \gamma \right) - \theta = \big( X(\gamma)^{\top} X(\gamma) \big)^{-1} X( \gamma )^{\top} \big[ u + \left( X_{\gamma_0} - X_{\gamma} \right) \delta_n \big]$.  

Hence, under Assumptions \ref{assumption1}-\ref{assumption4} and by the asymptotic results given by Lemma \ref{results}, we can show that the following hold
\begin{align*}
(i) &\ n^{\uptau} \big( X(\gamma)^{\top} X( \gamma ) \big)^{-1} X( \gamma )^{\top} u = n^{\uptau - 1} \left( \frac{1}{n^2} \sum_{t=1}^n X_t(\gamma)  X_t(\gamma)^{\top} \right)^{-1} \left( \frac{1}{n} \sum_{t=1}^n X_t(\gamma) u_t \right)
\\
&\Rightarrow n^{\kappa - 1} \mathcal{M} (\gamma)^{-1} \sigma_u \int_0^1 \mathbf{G}_{c , \phi }(s) dW \left( s, F( \gamma) \right) = \mathcal{O}_p(1),
\\
(ii) &\ n^{\kappa - 1}  \big( X(\gamma)^{\prime} X( \gamma ) \big)^{-1} X( \gamma )^{\top} \left( X_{\gamma_0} - X_{\gamma} \right) \delta_n 
= 
\left( \frac{1}{n^2} \sum_{t=1}^n X_t(\gamma)  X_t(\gamma)^{\top} \right)^{-1} 
\begin{bmatrix}
\displaystyle \frac{1}{n^2} \sum_{t=1}^n \left( x_{t, \gamma_0}^{\top} - x_{t, \gamma}^{\top} \right)\delta_0 
\\
\\
\displaystyle \frac{1}{n^2} \sum_{t=1}^n \left( x_{t, \gamma_0}^{\top} - x_{t, \gamma}^{\top} \right) \delta_0
\end{bmatrix}
\\
&\Rightarrow \mathcal{M}(\gamma)^{-1} \mathcal{A} \left( \gamma, \gamma_0, \delta_0 \right), 
\end{align*} 
where we use the fact that $I \left( q_t \leq \gamma \right) I \left( q_t \leq \gamma_0 \right) = I \left( q_t \leq \gamma \wedge \gamma_0 \right)$. 

\end{proof}

\newpage 

\underline{ \textbf{Proof of Lemma \ref{lemma1}}: }

\begin{proof}
We define the projection matrix $P_{\gamma} = X ( \gamma ) \big( X( \gamma )^{\top} X ( \gamma ) \big)^{-1} X ( \gamma )^{\top}$. Therefore, by simple calculation we have that 
\begin{align*}
SSR (\gamma) 
&:= 
y^{\top} \left( I_n - P( \gamma) \right) y
\\
&= 
\delta_n^{\top} X ( \gamma_0 )^{\top} \left( I_n - P( \gamma) \right)X ( \gamma_0 ) \delta_n + 2 \delta_n^{\top} X ( \gamma_0 )^{\top} \left( I_n - P( \gamma) \right) u + u^{\top} \left( I_n - P( \gamma) \right) u. 
\end{align*}
Hence, we considered the corresponding centred process around the neighbourhood of the true threshold parameter $\gamma_0$ such that 
\begin{align*}
\frac{1}{n^2} \big( SSR (\gamma) - SSR (\gamma_0) \big) 
&= \delta_n^{\top} X ( \gamma_0 )^{\top} \left( I_n - P( \gamma) \right)X ( \gamma_0 ) \delta_n 
\\
&+ 2 \delta_n^{\top} X ( \gamma_0 )^{\top} \left( I_n - P( \gamma) \right) u
\\
&+ u^{\prime} \left( I_n - P( \gamma) \right) u - u^{\top} \left( I_n - P( \gamma) \right) u 
\\
&= \mathcal{S}_{n1} + \mathcal{S}_{n2} + \mathcal{S}_{n3}.
\end{align*}
Then, for all $\gamma \in ( \gamma_0, \gamma_2 ]$ similar to the proof of Lemma A.5 of \cite{chen2015robust}, we can show 
\begin{align*}
\mathcal{S}_{n1} 
&= n^{2 - 2 \uptau} \delta_0^{\top} X(\gamma_0)^{\top} \left( I_n - P( \gamma) \right)X(\gamma_0) \delta_0 
\\
&\Rightarrow n^{- 2 \uptau} \big( F(\gamma_0) - F(\gamma_0)F(\gamma)^{-1}    F(\gamma_0) \big) \delta_0^{\top} \times  \int_0^1 \mathbf{G}_{c , \phi }(s) \mathbf{G}^{\top}_{c , \phi }(s) ds,  
\\
\\
\mathcal{S}_{n2} 
&= 2 n^{- 2 - \uptau} \delta_0^{\top} X(\gamma_0)^{\top} \left( I_n - P( \gamma) \right) X(\gamma_0) u 
\\
&\Rightarrow - 2 n^{- 1/ 2 - \uptau} \big( 1 - F(\gamma_0) F(\gamma)^{-1} \big) \times \int_0^1 \delta_0^{\top} \mathbf{G}_{c , \phi }(s) d W \left( s, F(\gamma_0) \right),
\\
\\
\mathcal{S}_{n3} 
&= n^{ - 2 } u^{\top} \left( P(\gamma_0) - P(\gamma) \right) u 
\\
&\Rightarrow \frac{1}{n} \sigma_u^2 \int_0^1 \mathbf{G}_{c , \phi }(s) d W \left( s, F(\gamma_0) \right) \left[ F(\gamma_0) \int_0^1 \mathbf{G}_{c , \phi }(s) \mathbf{G}^{\top}_{c , \phi }(s)  \right]^{-1} \times \int_0^1 \mathbf{G}^{\top}_{c , \phi }(s) d W \left( s, F(\gamma_0) \right) 
\\
&- \ \frac{1}{n} \sigma_u^2 \int_0^1 \mathbf{G}_{c , \phi }(s) d W \left( s, F(\gamma) \right) \left[ F(\gamma) \int_0^1 \mathbf{G}_{c , \phi }(s) \mathbf{G}^{\top}_{c , \phi }(s) \right]^{-1} \times \int_0^1 \mathbf{G}^{\top}_{c , \phi }(s) d W \left( s, F(\gamma) \right) 
\end{align*}
where we use the fact that $I \left( q_t \leq \gamma \right) I \left( q_t \leq \gamma_0 \right) = I \left( q_t \leq \gamma_0 \right)$.

Therefore, we have that 
\begin{align*}
 n^{ 2 \uptau - 2 } \big( SSR(\gamma) - SSR(\gamma_0) \big) \Rightarrow \big( F(\gamma_0) - F(\gamma_0)F(\gamma)^{-1} F(\gamma_0) \big)\delta_0^{\top} \int_0^1 \mathbf{G}_{c , \phi }(s) \mathbf{G}^{\top}_{c,\phi}(s)ds 
\end{align*}
Observing that $F(\gamma) > F(\gamma_0)$, we have that $n^{ 2 \uptau - 2 } \big( SSR(\gamma) - SSR(\gamma_0) \big) > 0$ uniformly for all $\gamma \in ( \gamma_0, \gamma_2 ]$. Similarly, we show that $n^{ 2 \uptau - 2 } \big( SSR(\gamma) - SSR(\gamma_0) \big) < 0$ if $\gamma \in \left[ \gamma_1, \gamma_0 \right)$, which completes the proof of the Theorem.  
\end{proof}

\newpage 

\underline{\textbf{Proof of Lemma \ref{lemma2}}:}

\begin{proof}
We follow the proof of Lemma A.6 of \cite{chen2015robust}, aiming to show that 
\begin{align}
\kappa_n \left( \widehat{\gamma} - \gamma_0 \right) = \underset{ v \in ( - \infty, + \infty ) }{ \text{arg min} } Q_n( \upupsilon ) = \mathcal{O}_p(1)
\end{align}
where $\kappa_n = n^{2 - 2 \uptau}$ and $Q_n( \upupsilon ) = \text{SSR} \left( \gamma_0 + \frac{ \upupsilon }{ \kappa_n }  \right) - \text{SSR} \left( \gamma_0 \right)$. Notice that 
\begin{align}
y - X( \gamma ) \widehat{\theta} = u + X \left( \beta - \widehat{\beta} \right) + X_{\gamma_0} \left( \delta - \widehat{\delta}_n \right) - \big( X_{\gamma} - X_{\gamma_0} \big) \widehat{\delta}_n. 
\end{align}  

Therefore, we can show that 
\begin{align*}
\text{SSR} &\left( \gamma \right) - \text{SSR} \left( \gamma_0 \right)
\\
&=
\widehat{\delta}_n^{\top} \big(  X_{\gamma} - X_{\gamma_0}  \big)^{\top} \big(  X_{\gamma} - X_{\gamma_0}  \big) \widehat{\delta}_n - 2 \left[ u + X \left( \beta - \widehat{\beta} \right) +  X_{\gamma_0} \left( \delta - \widehat{\delta}_n \right) \right]^{\top} \left[ \left( X_{\gamma} - X_{\gamma_0} \right) \widehat{\delta}_n \right]
\\
&= 
\widehat{\delta}_n^{\top} \sum_{t=1}^n x_t x_t^{\top} \ \big| I \left( q_t \leq \gamma \right) - I \left( q_t \leq \gamma_0 \right) \big| \widehat{\delta}_n 
\\
&- 2 \widehat{\delta}_n^{\top} \sum_{t=1}^n u_t \big( x_{t, \gamma} - x_{t, \gamma_0} \big)
\\
&- 2 \widehat{\delta}_n^{\top} \sum_{t=1}^n \big[ \left( \beta - \widehat{\beta} \right)^{\top} x_t + \left( \delta - \widehat{\delta}_n  \right) x_{t, \gamma_0} \big] \big( x_{t, \gamma} - x_{t, \gamma_0} \big)
\\
&+ \left( \widehat{\delta}_n + \delta_n \right)^{\top}  \sum_{t=1}^n x_t x_t^{\top} \ \big| I \left( q_t \leq \gamma \right) - I \left( q_t \leq \gamma_0 \right) \big| \left( \widehat{\delta}_n - \widehat{\delta} \right)
\\
&= \mathcal{S}_{n1} ( \gamma ) - 2 \mathcal{S}_{n2} ( \gamma ) - 2 \mathcal{S}_{n3} ( \gamma ) + \mathcal{S}_{n4} ( \gamma ). 
\end{align*}

Next, we consider the limiting behaviour of each of $\mathcal{S}_{nj} ( \gamma )$, for $j = 1,2,3,4$. Thus, for any $\upupsilon  \in \left[ \upupsilon_1 , \upupsilon_2 \right]$ the finite interval, using the Taylor expansion we have that 
\begin{align*}
\mathcal{S}_{n1} \left( \gamma_0 + \frac{ \upupsilon }{ \kappa_n } \right) 
&= 
n^{2 - 2 \uptau} \delta_0^{\top} \left\{ \left| F \left( \gamma_0 + \frac{ \upupsilon }{ \kappa_n } \right) - F \left( \gamma_0 \right) \right| \times \int_0^1 \mathbf{G}_{c , \phi }(s) \mathbf{G}^{\top}_{c,\phi}(s)ds \right\} \delta_0 + o_p(1),
\\
&\overset{ p }{ \to } | \upupsilon | f_{\gamma_0} \times \delta_0^{\prime} \left\{ \int_0^1 \mathbf{G}_{c , \phi }(s) \mathbf{G}^{\top}_{c,\phi}(s)ds \right\} \delta_0
\end{align*}

Moreover, for $S_{n2}$ and $\upupsilon$ we have that 
\begin{align*}
\mathcal{S}_{n2} \left( \gamma_0 + \frac{ \upupsilon }{ \kappa_n } \right) 
&= 
\delta_n^{\top} \sum_{t=1}^n u_t x_t \left[ I \left( q_t \leq \gamma_0 + \frac{ \upupsilon }{ \kappa_n } \right) - I \left( q_t \leq \gamma_0 \right)  \right] 
\\
&\ \ \ + \left( \widehat{\delta}_n - \delta_n \right)^{\top} \sum_{t=1}^n u_t x_t \left[ I \left( q_t \leq \gamma_0 + \frac{ \upupsilon }{ \kappa_n } \right) - I \left( q_t \leq \gamma_0 \right)  \right]  ,
\\
&\Rightarrow n^{2 - 2 \uptau} \delta_0^{\top} \sigma_u \int_0^1 \mathbf{G}_{c , \phi }(s) d \left\{ W \left( s, F \left( \gamma_0 + \frac{ \upupsilon }{ \kappa_n }  \right) \right) - W \left( s, F \left( \gamma_0 \right) \right)  \right\} \left[ 1 + \mathcal{O}_p( n^{- \kappa} ) \right]. 
\end{align*} 

\newpage 

Define $\mathbf{H} ( m ) := \displaystyle \int_0^1 \mathbf{G}^{\top}_{c,\phi}(s) dW (s, m)$. Furthermore, notice that $\mathbb{E} \big( \mathbf{H} ( m ) \big) = 0$ and $\displaystyle \text{Var} \big( \mathbf{H} ( m ) \big) =  m \int_0^1 \mathbf{G}_{c,\phi}(s) \mathbf{G}^{\top}_{c,\phi}(s) ds$ and $\displaystyle \mathbb{E} \big( \mathbf{H} ( m_1 ) \mathbf{H} ( m_2 ) \big) = \left( m_1 \wedge m_2 \right) \int_0^1 \mathbf{G}_{c,\phi}(s) \mathbf{G}^{\top}_{c,\phi}(s) ds$. 

Therefore, we define the Brownian motion $\widetilde{\mathbf{H}}(\upupsilon)$ which is defined as 
\begin{align}
\widetilde{\mathbf{H}}(\upupsilon) := \kappa_n^{1 / 2} \delta_0^{\top} \sigma_u \left[ \mathbf{H} \left( \gamma_0 + \frac{ \upupsilon }{ \kappa_n } \right) -  \mathbf{H} \left( \gamma_0 \right) \right] 
\end{align}
with stochastic variance function given by 
\begin{align}
\sigma^2_{\widetilde{\mathbf{H}}(\upupsilon) } :=
\text{Var} \big( \widetilde{\mathbf{H}}(\upupsilon) \big) 
\equiv 
\sigma_u^2 f( \gamma_0 )  \upupsilon \delta_0^{\top} \left[ \int_0^1 \mathbf{G}_{c,\phi}(s) \mathbf{G}^{\top}_{c,\phi}(s) ds \right]\delta_0
\end{align}
using the first-order Taylor expansion of $F \left( \gamma_0 + \frac{ \upupsilon }{ \kappa_n } \right)$ around $\gamma_0$. Moreover, we can easily show that for all $\upupsilon  \in \left[ \upupsilon_1 , \upupsilon_2 \right]$, $\mathcal{S}_{n3} ( \upupsilon ) = \mathcal{O}_p \left( \mathcal{S}_{n1} ( \upupsilon ) \right)$ and $\mathcal{S}_{n4} ( \upupsilon ) = \mathcal{O}_p \left( S_{n1} ( \upupsilon ) \right)$. Therefore, by combining the above results and by denoting $\displaystyle \mathcal{K}_{c, \phi}( \delta_0) := \delta_0^{\top} \left[ \int_0^1 \mathbf{G}_{c,\phi}(s) \mathbf{G}^{\top}_{c,\phi}(s) ds \right] \delta_0 $ for notational convenience, we obtain that 
\begin{align*}
\text{SSR} &\left( \gamma_0 + \frac{ \upupsilon }{ \kappa_n } \right) - \text{SSR} \left( \gamma_0 \right)
&\Rightarrow  
\begin{cases}
\displaystyle  \upupsilon f \left( \gamma_0 \right) \mathcal{K}_{c, \phi}( \delta_0) - 2 \sigma_u \sqrt{ \upupsilon f \left( \gamma_0 \right) \mathcal{K}_{c, \phi} ( \delta_0) } W_1 (- \upupsilon  ), \ \text{if} \ \upupsilon > 0,  
\\
\\
0, \ \text{if} \ \upupsilon = 0,  
\\
\\
\displaystyle | \upupsilon | f \left( \gamma_0 \right) \mathcal{K}_{c, \phi}( \delta_0) - 2 \sigma_u \sqrt{ | \upupsilon | f \left( \gamma_0 \right) \mathcal{K}_{c, \phi}( \delta_0) } W_2 (- \upupsilon  ), \ \text{if} \ \upupsilon < 0, 
\end{cases}
\end{align*}
where $W_1( \upupsilon )$ and  $W_2( \upupsilon )$ are two independent standard Brownian motions on $[0, \infty )$.
\end{proof}

\newpage

\underline{ \textbf{Proof of Theorem \ref{theorem1}}: }

\begin{proof}
We proceed with proving Theorem 1 of the paper.  Notice for a given $\gamma \in \left[ \gamma_1, \gamma_2 \right]$, the estimator of  $\delta$ is given by $\widehat{\delta} = \big[ X_{\gamma}^{\top} \left( I_n - P_x \right) X_{\gamma} \big]^{-1} X_{\gamma}^{\top} \left( I_n - P_x \right) y$. Hence, 
\begin{align}
\mathcal{W}_n ( \gamma ) 
&= 
\widehat{ \delta }^{\top} \big[ X_{\gamma}^{\top} \left( I_n - P_x \right) X_{\gamma} \big]^{-1} \widehat{ \delta }^{\top} \widehat{\sigma}^{-2}
\nonumber
\\
&= 
y^{\top} \left( I_n - P_x \right) X_{\gamma} \big[ X_{\gamma}^{\top} \left( I_n - P_x \right) X_{\gamma} \big]^{-1}  X_{\gamma}^{\top} \left( I_n - P_x \right) y
\end{align}
Therefore, by Lemma 2 we can show that
\begin{align*}
(i) \ & n^{-2} X_{\gamma}^{\top} \left( I_n - P_x \right) \Rightarrow \big[ F( \gamma ) - F( \gamma )^2 \big] \int_0^1 \mathbf{G}_{c,\phi}(s) \mathbf{G}^{\top}_{c,\phi}(s) ds, 
\\
(ii) \ & n^{-1} X_{\gamma}^{\top} \left( I_n - P_x \right) y = n^{-1} \big[ X_{\gamma}^u - X_{\gamma}^{\top} P_x u +  X_{\gamma}^{\top} X_{\gamma_0 } \delta_n - X_{\gamma}^{\top} P_x X_{\gamma_0 } \delta_n  \big]
\nonumber 
\\
&\Rightarrow  \int_0^1 \mathbf{G}_{c,\phi}(s) dW \left( s, F(\gamma ) \right)  
- F( \gamma ) \int_0^1 \mathbf{G}_{c,\phi}(s) dW \left( s \right)  
\nonumber
\\
&+ n^{1 - \uptau } F( \gamma  \wedge \gamma_0 ) \left[ \int_0^1 \mathbf{G}_{c,\phi}(s) \mathbf{G}^{\top}_{c,\phi}(s) ds \right] \delta_0 -  n^{1 - \uptau } F( \gamma ) F( \gamma_0 ) \left[ \int_0^1 \mathbf{G}_{c,\phi}(s) \mathbf{G}^{\top}_{c,\phi}(s) ds \right] \delta_0.
\end{align*}

Therefore, we can conclude that $\mathcal{W}_n ( \gamma ) = \mathcal{O}_p \left( n^{1 - \uptau }  \right)$. Under the null hypothesis of $\delta_0 = 0$, we can show that 
\begin{align*}
\mathcal{W}_{n,OLS}^{*} \Rightarrow \underset{ \gamma \in \left[ \gamma_1, \gamma_2 \right]   }{ \text{sup} } &\left\{ \int_0^1 \mathbf{G}_{c , \phi }(s)  dW \left( s, F( \gamma) \right) - F( \gamma) \int_0^1 \mathbf{G}_{c , \phi }(s) dW(s) \right\}^{\top}
\\
&\times \bigg\{ \big[ F(\gamma) - F(\gamma)^2 \big] \int_0^1 \mathbf{G}_{c , \phi }(s) \mathbf{G}^{\top}_{c , \phi }(s) ds \bigg\}^{-1} 
\\
&\left\{ \int_0^1 \mathbf{G}_{c , \phi }(s)  dW \left( s, F( \gamma) \right) - F( \gamma) \int_0^1 \mathbf{G}_{c , \phi }(s) dW(s) \right\} \big/ \sigma^2_u
\end{align*}
which completes the proof of the result given by Theorem \ref{theorem1}. 
\end{proof}

\medskip



\underline{ \textbf{Proof of Theorem \ref{theorem2}}: }

Part (i): We prove that the limit theory of the sup Wald-OLS statistic under the null hypothesis of linearity and no predictability has a non-standard form.

Part (ii): We prove that the limit theory of the sup Wald-IVX statistic under the null hypothesis of linearity and no predictability has the familiar form as in Proposition 2 of \cite{gonzalo2012regime} which is decomposed into a standard $\chi^2$ component and a standard Brownian Bridge term.

\newpage 

\paragraph{Proof of Lemma 1}

\begin{proof}
Consider the autoregressive model with order one and STLUR
\begin{align}
Y_t = \mathsf{exp} \left\{ \frac{c}{n} + \frac{ \varphi^{\prime} v_t }{ \sqrt{n} } \right\} Y_{t-1} + u_t, \ \ \ t = 1,...,n.    
\end{align}
\begin{align}
Y_t = \sum_{j=1}^t \mathsf{exp} \left\{  \frac{ (t-j)c}{n} -  \frac{ \displaystyle\boldsymbol{\varphi}^{\prime} \sum_{i = j+1}^t v_i }{ \sqrt{n} }  \right\} u_j, t \geq 2    
\end{align}

\end{proof}

\newpage

\end{document}